\documentclass[11pt]{article}
\usepackage{amsthm}
\usepackage{amsmath}
\usepackage{amssymb}
\usepackage{refcount}
\usepackage{fullpage}
\usepackage{paralist}
\usepackage{appendix}
\usepackage{graphicx}
\usepackage{bm}
\usepackage[ruled,boxed,vlined]{algorithm2e}
\DontPrintSemicolon
\usepackage{verbatim}

\renewcommand{\phi}{\varphi}

\newcommand{\hide}[1]{ }

\renewcommand{\mathbf}{\bm}

\theoremstyle{plain}
\newtheorem{theorem}{Theorem}[section]

\newtheorem{lemma}[theorem]{Lemma}

\newtheorem{definition}{Definition}
\newtheorem*{remark}{Remark}
\newtheorem{observation}{Observation}
\renewcommand{\include}{\input}
\usepackage{tikz}
\usetikzlibrary{decorations.pathreplacing,angles,quotes}
\usetikzlibrary{arrows}
\usepackage{float}
\usetikzlibrary{patterns}

\newcommand{\algmub}{\text{ALG}_{\text{MUB}}}
\newcommand{\algrev}{\text{ALG}}
\newcommand{\mub}{\text{MUB}}

\title{Hierarchical Clustering: a 0.585 Revenue Approximation}
\author{Noga Alon%
	\thanks{Princeton University and Tel Aviv University. 
Email: nalon@math.princeton.edu. Research supported in part by
NSF grant DMS-1855464, 
BSF grant 2018267
and the Simons Foundation.
}
	\and
	Yossi Azar%
	\thanks{School of Computer Science, Tel-Aviv University. Email: azar@tau.ac.il. Research supported in part by the Israel Science Foundation (grant No. 1506/16)}
	\and
	Danny Vainstein%
	\thanks{School of Computer Science, Tel-Aviv University. Email: dannyvainstein@gmail.com}
}

\begin{document}

\maketitle

\begin{abstract}
Hierarchical Clustering trees have been widely accepted as a useful form 
of clustering data, resulting in a prevalence of adopting fields
including phylogenetics, image analysis, bioinformatics and more. 
Recently, Dasgupta (STOC 16') initiated the analysis of 
these types of algorithms through the lenses of approximation. 
Later, the dual problem was considered by Moseley and Wang (NIPS 17') dubbing it the Revenue goal function. 
In this problem, given a nonnegative weight $w_{ij}$ for each pair 
$i,j \in [n]=\{1,2, \ldots ,n\}$, the objective is to find a tree $T$ whose
set of leaves is $[n]$ that maximizes the function 
$\sum_{i<j \in [n]} w_{ij} (n -|T_{ij}|)$, where $|T_{ij}|$  is the number
of leaves in the subtree rooted at the least common ancestor of $i$ and $j$.

In our work we consider the revenue goal function and prove the following 
results. First, we prove the existence of a bisection 
(i.e., a tree of depth $2$ in which the root has two children, each being
a parent of $n/2$ leaves) 
which approximates the general optimal tree solution up to a factor of 
$\frac{1}{2}$ (which is tight). Second, we apply 
this result in order to prove a $\frac{2}{3}p$ approximation for the general revenue problem, where $p$ is defined as the approximation ratio of the \textsc{Max-Uncut Bisection} problem. 
Since $p$ is known to be at least $0.8776$ \cite{Better_Balance_by_Being_Biased_A_08776_Approximation_for_Max_Bisection} \cite{An_improved_semidefinite_programming_hierarchies_rounding_approximation_algorithm_for_maximum_graph_bisection_problems}, 
we get a $0.585$ approximation algorithm for the revenue problem. This 
improves a sequence of earlier results which culminated in an
$0.4246$-approximation guarantee \cite{Bisect_and_Conquer:_Hierarchical_Clustering_via_Max-Uncut_Bisection}.
\end{abstract}

%\thispagestyle{empty}

%\pagebreak

%\setcounter{page}{1} 

\section{Introduction}

The notion of Hierarchical Clustering (HC) trees has been introduced and 
subsequently studied for several decades. The notion was first 
considered due to its applications to the realm of phylogenetics \cite{Numerical_taxonomy} \cite{A_model_for_taxonomy}. Here, given genomic similarities between species our goal is to hierarchically cluster the species in a way that captures 
fine-grained relations between different species. Since then, the notion of 
HC trees has expanded to many additional fields. See
\cite{A_Survey_of_Clustering_Data_Mining_Techniques}
for a survey on the subject.

Typically, schemes for generating HC trees fall into one of two categories: Agglomerative algorithms (i.e., bottom-up) and Divisive algorithms (i.e., top-down). 
Agglomerative algorithms initially start with a partition in which 
each data point forms its own set. The algorithm then proceeds to recursively merge different sets, terminating once all data points are contained in the same set. Notably, the well-studied average linkage algorithm 
is an example of such a procedure in which in each step 
two sets maximizing the average induced weight are merged. 
On the other hand, divisive algorithms start with a single set 
containing all the data points, and then proceed to recursively 
split sets, terminating once each data point remains alone in its set.

Recently, Dasgupta \cite{a_cost_function_for_similarity-based_hierarchical_clustering} 
formally defined the notion of a "good" HC tree. With this definition 
he elegantly bridged the gap between HC trees and the field of 
approximation algorithms by defining a minimization cost function (see related work for an extended discussion on this goal function). Thereafter, Moseley and Wang \cite{Approximation_Bounds_for_Hierarchical_Clustering:_Average_Linkage} considered the complementary 
%(being that the optimal solutions of the two are one and the same) 
maximization variant of this problem - namely, the revenue goal function.

Both problems considered the following model. Assume we are given a set of $n$ data points, $V$ with some notion of similarity between them. The similarity is formally captured through similarity-edges between 
any two data points, represented by  a weighted graph
$G = (V,E,w)$ with $w:E \rightarrow \mathbb{R}^+$. Alternatively, 
$G$ may be viewed as a complete graph in which $w(e) = 0$ 
for any $e \not \in E$. Our goal is then to construct an HC tree, $T$, such that its leaves are in a 1-1 correspondence with our data points, $V$. 
Note that given such a tree, every internal node represents 
a set of data points, which are the leaves of its subtree,
and a partition of this set given by the sets of leaves of the
subtrees rooted at the children of the node.

Intuitively speaking, since higher weighted edges correspond 
to more similar data points, it is desirable to split 
the endpoints of such edges low in a good
HC tree. 
Formally, Moseley and Wang \cite{Approximation_Bounds_for_Hierarchical_Clustering:_Average_Linkage} defined the revenue problem as,

\begin{equation*}
\textstyle\max_{T} \, R_G(T) = \textstyle\max_{T} \, \sum_{e=\{ij\} \in E} w_{ij} (n-|T_{ij}|),    
\end{equation*}
where $T$ is an HC tree, $T_{ij}$ denotes the subtree rooted at the least-common-ancestor (LCA) of data points $i$ and $j$ and $|T_{ij}|$ denotes the amount of data points in $T_{ij}$. 

In \cite{Approximation_Bounds_for_Hierarchical_Clustering:_Average_Linkage}, Moseley and Wang considered several algorithms. Notably, they considered the random algorithm, that simply splits data points randomly at every cut (henceforth denoted by $RAND$), and the average-linkage algorithm. They showed that both yield an approximation factor of 1/3. Subsequently, Charikar et al. \cite{Hierarchical_Clustering_better_than_Average_Linkage} showed that one can beat average-linkage through the use of semi-definite programming, improving the bound to 0.3364. Recently, Ahmadian et al. 
\cite{Bisect_and_Conquer:_Hierarchical_Clustering_via_Max-Uncut_Bisection} managed to leverage the \textsc{Max-Uncut Bisection} ($\mub$) problem in order to prove a 0.4246 approximation. In our paper we improve upon this result and show an improved approximation of 0.585. \\

%Specifically, the MUB problem is that of finding a bisection (i.e., a partition of size 2 with 2 equal sets) of a graph, such that the weight of uncut edges, is maximized. Ahmadian et al. \cite{Bisect_and_Conquer:_Hierarchical_Clustering_via_Max-Uncut_Bisection} considered running two algorithms simultaneously and choosing that which maximizes the revenue. The first algorithm is simply the random algorithm, $RAND$. The second algorithm first performs a cut that aims at maximizing the MUB problem, and then continues randomly. They then leverage the fact that exists a $p = 0.8776$ approximation for the MUB problem (\cite{Better_Balance_by_Being_Biased:_A_0.8776_Approximation_for_Max_Bisection} \cite{An_improved_semidefinite_programming_hierarchies_rounding_approximation_algorithm_for_maximum_graph_bisection_problems}) in order to show that the combined algorithm is a 0.4246 approximation. In our paper we show that in fact, their second algorithm alone is enough to yield a $0.585$ approximation.\\

\noindent \textbf{Our contributions.} We consider the revenue goal function and prove the following results.
\begin{itemize}
    \item We show that for any revenue instance,
there exists a bisection, $X$, that is, a tree of depth $2$ in which the root
has two children, each being the parent of $n/2$ leaves,
such that $R(X) \geq \frac{1}{2}OPT$, where $OPT$ denotes the revenue gained by the optimal tree (see Theorem \ref{theorem.0.5_bisection}). In order to show such existence we make use of two random processes: 
we randomly fix the order of the leaves in the optimal tree 
in an appropriate way,
and then randomly generate our bisection, $X$. 
We emphasize the fact that even though $OPT$ makes use of an arbitrarily deep tree, it is enough to consider a single cut in order to gain half the revenue. 

    \item Using our result regarding the existence of a large revenue generating bisection, we prove a 0.585 approximation for the revenue problem. We note that in fact we show a $\frac{2}{3}p$ approximation where $p=0.8776$ is the best known approximation for the $\mub$ problem.

\end{itemize}

\begin{remark}
\label{remark.our_algorithm_compared_to_ahmadin}
The algorithm we consider is that which solves the $\mub$ problem
for the first cut and then proceeds using the random algorithm. In
\cite{Bisect_and_Conquer:_Hierarchical_Clustering_via_Max-Uncut_Bisection},
Ahmadian et al. considered this algorithm coupled with the random
algorithm - they used both algorithms simultaneously while taking
the maximal revenue of the two. They showed that this results in a
0.4246 approximation with respect to the gained revenue. Somewhat 
surprisingly we show that the former algorithm ($\mub$ and then random) on its own is enough to yield an approximation of 0.585. 
\end{remark}

\noindent \textbf{Techniques.} Our first result makes use of a new 
upper bound on the optimal solution, which may be of independent interest. 
Specifically, given an optimal solution, we embed its leaves on a line such that its root is above the line and we have no resulting crossing edges. This clearly yields an ordering of the leaves (see Figure \ref{figure.example_of_tree_embedding}). We then consider the distance between any two data points, $i$ and $j$ within this ordering (simply the difference in rank) and observe that this is in fact a lower bound on $|T_{ij}|$. This in turn yields an upper bound on the optimal solution.

Next, we make use of this bound by randomly generating a bisection that 
gains revenue that is "large" with respect to the bound and 
by showing that in expectation this bound is "far" from the 
optimal solution. Both "large" and "far" will be formally defined later on. \\

%\item Since exists a bisection with large revenue (compared to the optimal solution), the algorithm presented by \cite{Bisect_and_Conquer:_Hierarchical_Clustering_via_Max-Uncut_Bisection} seems to naturally fit here. This is due to the fact that this algorithm's gained revenue may be bounded by the revenue gained by the first cut. Thereafter, the first cut may be used to approximate the MUB problem, which by Theorem \ref{theorem.0.5_bisection} approximates the optimal solution. Overall, if the first cut approximates the MUB problem within a factor of $p$, then the entire tree approximates the revenue problem with a factor of $\frac{2p}{3}$. In particular, for $p = 0.8776$ we get an approximation of $0.585$.

\noindent \textbf{Related work.} Dasgupta \cite{a_cost_function_for_similarity-based_hierarchical_clustering} kicked off the line of work considering HC trees within the realm of approximation algorithms. In his paper, he considered similarity-edges and defined the cost of an HC tree as,
\begin{equation*}
\textstyle\min_T \, C_G(T) = \textstyle\min_T \, \sum_{e=\{ij\} \in E} w_{ij} |T_{ij}|.    
\end{equation*}

Note that the revenue goal function is in fact complementary to that of 
Dasgupta's function (that is, 
the optimal solution is the same for both, albeit with different goal function values).

In \cite{a_cost_function_for_similarity-based_hierarchical_clustering}, 
many general properties pertaining to this goal function were discussed. 
Notably, it was shown that this goal function is intuitive in that 
(1) on complete graphs (with no structure) all HC trees yield the 
same cost, (2) on disconnected graphs, optimal HC trees begin 
by splitting disconnected components and (3) the goal function is
modular. He further presented an $O(\log^{1.5} (n))$ approximation 
algorithm via recursive sparsest cut. Later, both Charikar and Chatziafratis
\cite{Approximate_Hierarchical_Clustering_via_Sparsest_Cut_and_Spreading_Metrics} and Cohen-Addad et al.  \cite{Hierarchical_Clustering:_Objective_Functions_and_Algorithms} showed that this algorithm 
is in fact an $O(\sqrt{\log n})$ approximation. 
In the hardness domain, Dasgupta \cite{a_cost_function_for_similarity-based_hierarchical_clustering} showed that the problem is NP-hard via 
a reduction to a variant of the NAE-SAT problem. 
This was later improved by Charikar and Chatziafratis \cite{Approximate_Hierarchical_Clustering_via_Sparsest_Cut_and_Spreading_Metrics}, 
showing that in fact no constant approximation exists (assuming the Small Set Expansion hypothesis). Cohen-Addad et al. \cite{Hierarchical_Clustering_Beyond_the_Worst-Case} 
managed to overcome the latter worst-case specific result by 
considering average case inputs defined  by a stochastic block model 
and its hierarchical extension. Here, they managed to establish
an $O(1)$ approximation.

Following Dasgupta's work, Cohen-Addad et al. \cite{Hierarchical_Clustering:_Objective_Functions_and_Algorithms} considered the case of dissimilarity-edges. In this case, Dasgupta's cost function is now translated to a maximization problem. Given an HC tree, $T$, we now denote its gained value as its gained dissimilarity, $D_G(T)$. In this setting, both the random algorithm ($RAND$) and the average-linkage algorithm yield dissimilarity values of $2/3$ of the optimal solution. Charikar et al. \cite{Hierarchical_Clustering_better_than_Average_Linkage} improved upon this by proving an approximation of 0.6671 which makes use of a more delicate multi-phase algorithm.

Several other extensions to the formerly defined HC goal functions were also considered. One such extension is that of structural constraints. Specifically, every constraint appears in the form of $"i,j|k"$ for data points $i$, $j$ and $k$. A constraint is then considered satisfied if $k \not \in T_{ij}$, for an HC tree, $T$. Aho et al. \cite{Inferring_a_tree_from_lowest_common_ancestors_with_an_application_to_the_optimization_of_relational_expressions} considered this problem in the phylogenetic realm where this notion gives rise to the problem of constructing a phylogenetic tree that satisfies a set of lineage constraints on common
ancestors. This notion has more recently been investigated 
in the domain of HC, by Chatziafratis et al.
\cite{Hierarchical_Clustering_with_Structural_Constraints}. In
their paper they extended Dasgupta's goal function to include
structural constraints and showed an $O(k \sqrt{\log n})$
approximation where $k$ is the number of constraints. For
additional works in the realm of HC trees see \cite{Clustering_with_interactive_feedback}, \cite{Local_algorithms_for_interactive_clustering}, \cite{Interactive_bayesian_hierarchical_clustering}, \cite{Hierarchical_Clustering_for_Euclidean_Data}.

\section{Notation and Preliminaries}

We introduce several definitions that will aid us throughout the following sections.

\begin{definition}
Given a revenue instance $G = (V,E,w)$, we denote the optimal revenue tree by $OPT$.
\end{definition}

We shall abuse notation and refer both to the optimal revenue tree and 
to its generated revenue as $OPT$ (it will be clear from context which of these definitions we will be referring to). 

\begin{definition}
Given an HC tree $T$, we denote by $R(T)$ the revenue it yields and for any similarity edge $e$, we denote by $R_T(e)$ the revenue gained by the edge $e$ with respect to the HC tree, $T$.
\end{definition}

The following definition will be useful when considering the revenue generated with respect to some similarity edge.

\begin{definition}
Given a revenue instance $G = (V,E,w)$ and a similarity edge $e = \{i,j\}$ we denote $T_e = T_{ij}$.
\end{definition}

Note that therefore, for any HC tree $T$, $R_T(e) = w_e(n -
|T_e|)$, where $|T_e|$ is the number of leaves in the tree
$T_e$. Further note that with these 
definitions we may assume w.l.o.g. that any optimal tree is binary. 

Finally, to ease notation we shall refer to the data points of a given revenue instance $G = (V,E,w)$, as $V = \{1,2,\ldots, n\}$.

% Given a revenue instance, $G = (V,E,w)$ where $V = \{1,2,\ldots, n\}$, we denote by $OPT$ both the optimal revenue tree and its yielded revenue (it will be clear from context which of these definitions we will be referring to). Given an HC tree, $T$, we denote by $R(T)$ the revenue it yields and for any similarity edge, $e$, we denote by $R_T(e)$ the revenue gained by the edge $e$ with respect to the HC tree, $T$ (i.e., $R_T(e) = w_e(n - |T_e|)$). Note that under that under such definitions we may assume w.l.o.g. that any optimal tree is binary.

\section{Existence of a High Revenue Bisection}

\begin{theorem}
\label{theorem.0.5_bisection}
For any revenue instance and corresponding optimal solution, $OPT$, 
there exists a bisection, $X$, satisfying
\[
R(X) \geq \frac{1}{2}OPT,
\]
and this is asymptotically tight.
\end{theorem}

\noindent Before we prove the theorem, we introduce some notation. 
To simplify the presentation we assume  from now on that the number
$n$ of leaves is even. \\

\noindent \textbf{Defining a tree ordering:} Given an HC tree $T$, we may fix its leaves in several different orders. Each order is produced by representing $T$  as a planar graph where its leaves are all embedded on a line, the root is above the line and all edges of the tree are straight lines going down from each parent to its children with no crossing edges. Denote each such ordering by the function $\pi(T) : V \rightarrow [n]$.\\

Henceforth we will consider $\pi$ with respect to the optimal solution $OPT$ and denote $\pi = \pi(OPT)$. Recall that we may assume w.l.o.g. that $OPT$ is a binary tree. Next we define a distribution over all feasible orderings, $\pi$.

\begin{definition}
Let $P_{\pi}$ denote the distribution generated by choosing uniformly at random a subset of internal nodes from $OPT$ and swapping the placements of the left and right subtrees of each of these chosen nodes.
\end{definition}

% Before we prove the theorem, we introduce some notation. Recall that we may assume w.l.o.g. that the optimal solution is a binary tree. Given such an optimal tree, $OPT$, we may fix its leaves in several 
% different orders. Each order is produced by representing $OPT$ 
% as a planar 
% graph where its leaves are all embedded on a line, the root is above the 
% line and all edges of the tree are straight lines going down from each parent
% to its children with no crossing edges.
% Denote each such ordering by the function $\pi : V \rightarrow [n]$. Note that a simple way of producing a uniform distribution on all possible orderings is to fix one ordering, choose uniformly at random a subset of internal nodes and swap the placements of the left and right subtrees of each of these chosen nodes.

Finally, given an ordering $\pi$, and an edge $e = \{i,j\}$, let $y_e^\pi = |\pi(i) - \pi(j)|$, denote the distance between leaves $i$ and $j$ in the fixed tree, $OPT$. For a pictorial example, see Figure \ref{figure.example_of_tree_embedding}.

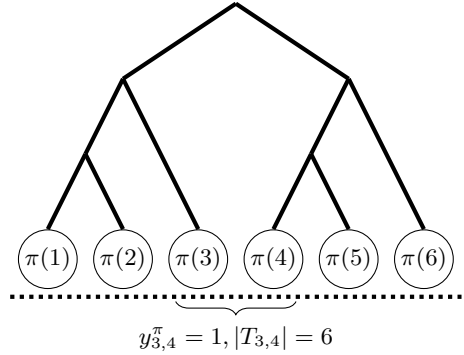
\begin{figure}[H]
\centering
\begin{tikzpicture}

\draw [ultra thick] (0,0) -- (0.5, 1);
\node [circle,draw=black,inner sep=1pt] (d_1) at (0,0-.4) {\footnotesize $\pi(1)$};
\draw [ultra thick] (1,0) -- (0.5, 1);
\node [circle,draw=black,inner sep=1pt] (d_2) at (1,0-.4) {\footnotesize $\pi(2)$};
\draw [ultra thick] (0.5,1) -- (1,2);
\draw [ultra thick] (2,0) -- (1,2);
\node [circle,draw=black,inner sep=1pt] (d_3) at (2,0-.4) {\footnotesize $\pi(3)$};
\draw [ultra thick] (1,2) -- (2.5,3);

\draw [ultra thick] (3,0) -- (3.5,1);
\node [circle,draw=black,inner sep=1pt] (d_4) at (3,0-.4) {\footnotesize $\pi(4)$};
\draw [ultra thick] (4,0) -- (3.5,1);
\node [circle,draw=black,inner sep=1pt] (d_5) at (4,0-.4) {\footnotesize $\pi(5)$};
\draw [ultra thick] (3.5,1) -- (4,2);
\draw [ultra thick] (5,0) -- (4,2);
\node [circle,draw=black,inner sep=1pt] (d_6) at (5,0-.4) {\footnotesize $\pi(6)$};
\draw [ultra thick] (4,2) -- (2.5,3);

\draw [ultra thick, dotted] (-.5,-.9) -- (5.5,-.9);

\draw [decorate,decoration={brace,amplitude=5pt,mirror,raise=4ex}]
  (2-.3,-.3) -- (3+.3,-.3) node[midway,yshift=-3em]{\footnotesize $y_{3,4}^\pi=1, |T_{3,4}|=6$};

\end{tikzpicture}
\caption{Embedding an HC tree with order $\pi$.}
\label{figure.example_of_tree_embedding}
\end{figure}

\noindent In order to prove Theorem \ref{theorem.0.5_bisection} we first upper bound the revenue gained by $OPT$.

\begin{observation}
$\forall \pi, \forall e\in E: y_e^\pi \leq |T_e|$, where both $\pi$ and $T_e$ are defined with respect to $OPT$.
\end{observation}

\noindent In itself, this will not result in a good enough bound on $OPT$. Therefore, we show that on average, $y_e^\pi$ is far from $T_e$.

\begin{lemma}
\label{lemma.y_is_far_from_real_value}
$E_{P_{\pi}} [y_e^\pi] = \frac{|T_e|}{2}$.
\end{lemma}

\begin{proof}
Let $e = \{i,j\}$. It can be shown through simple induction that 
for any tree with $n$ leaves, and for any leaf, 
$k$, $E_{P_{\pi}} [\pi(k)] = \frac{n+1}{2}$.  (This also follows by
linearity of expectation from the simple fact that for every two
distinct leaves $i$ and $j$ the probability that $\pi(i)<\pi(j)$
is exactly $1/2$.)

Denote by $n_i$ and $n_j$ the number of leaves in the subtrees containing $i$ and $j$, each rooted at a separate child of $i$ and $j$'s least-common-ancestor. 
By the definition of $P_{\pi}$, the probability that $i$ (together with the $n_i$ leaves of its subtree) appears before $j$ (and the $n_j$ leaves in its subtree) is exactly $1/2$, resulting in:
\[
E_{P_{\pi}} [y_e^\pi] = (1/2)(n_i + \frac{n_j+1}{2} - \frac{n_i+1}{2}) + (1/2)(n_j + \frac{n_i+1}{2} - \frac{n_j+1}{2}) = \frac{|T_e|}{2}.
\]
\end{proof}

\noindent Let $Y_\pi = \sum_{e \in E} w_e \cdot y_e^\pi$. By 
linearity of expectation we get the following lemma.

\begin{lemma}
\label{lemma.small_weighted_ordering}
There exists an ordering of $OPT$, $\pi^*$, such that,
\[
Y_{\pi^*} \leq \sum_{e \in E} w_e \cdot \frac{|T_e|}{2}.
\]
\end{lemma}

\begin{proof}
By linearity of expectation and Lemma \ref{lemma.y_is_far_from_real_value},
\[
E_{P_{\pi}} [Y_\pi] = \sum_{e \in E} w_e \cdot E_{P_{\pi}}[y_e^\pi] = \sum_{e \in E} w_e \cdot \frac{|T_e|}{2}.
\]
Therefore, there exists an ordering as needed.
\end{proof}

\noindent Next we show that there exists a distribution over the
set of all bisections with high (to some degree) revenue with
respect to the revenue gained by considering $y_e^\pi$ rather than
$|T_e|$.

\begin{lemma}
\label{lemma.high_revenue_distribution_of_bisections}
Given any ordering of $OPT$, $\pi$, there 
exists a distribution, $P_X$, over all bisections, $X$, such that for any edge $e \in E$,
\[
E_{P_X} [R_X(e)] \geq  \frac{1}{2} w_e (n - 2y_e^\pi).
\]
\end{lemma}

\begin{proof}
Fix $\pi$. We relabel the leaves of $OPT$ such that $\pi(i) = i$ (this is 
simply to ease the notation). Next we define a distribution over all bisections, $X$. Consider the following random process: choose $x$ uniformly at random from $\{1, \ldots, \frac{n}{2}\}$. Then define $X$ to be the bisection, $\{x, x+1, \ldots, x+\frac{n}{2} - 1\}, \{1,2,\ldots, x-1, x+\frac{n}{2}, x+\frac{n}{2}+1, \ldots, n\}$.

Now consider some edge, $e \in E$. If $y_e^\pi \geq \frac{n}{2}$, 
then $n - 2y_e^\pi \leq 0$. $R_X(e)$ is always nonnegative 
(since it is defined as the revenue gained by $e$) and thus the 
assertion of the lemma holds in this case.

Otherwise, $y_e^\pi \leq \frac{n}{2}-1$ for $e = \{i,j\}$. In this case, the probability that $i$ and $j$ are cut by the bisection is 
exactly $ \frac{y_e^\pi}{n/2} = \frac{2y_e^\pi}{n}$. Furthermore, since $X$ is a bisection, any uncut edge yields a revenue of $n/2$. Overall,
\[
E_{P_X} [R_X(e)] \geq w_e(1 - \frac{2y_e^\pi}{n})\frac{n}{2},
\]
completing the proof of the lemma.
\end{proof}

Recall that by the definition of $R(X)$ it follows that 
$R(X) = \sum_{e \in E} R_X(e)$. Therefore, by combining lemmas \ref{lemma.small_weighted_ordering} and \ref{lemma.high_revenue_distribution_of_bisections}, we may sum over all edges and get the following observation.

\begin{observation}
\label{observation.expected_rev_gained_by_offline_bisection}
Given the ordering, $\pi^*$, guaranteed by Lemma \ref{lemma.small_weighted_ordering}, we get,
\[
E_{P_X}[R(X)] \geq \sum_{e \in E} \frac{1}{2} w_e (n - 2y_e^{\pi^*}).
\]
\end{observation}

\noindent We are now ready to prove Theorem \ref{theorem.0.5_bisection}.\\

\begin{proof}[Proof of Theorem \ref{theorem.0.5_bisection}]
Fix an ordering of $OPT$ to be $\pi^*$ as guaranteed by 
Lemma \ref{lemma.small_weighted_ordering}. By 
Observation \ref{observation.expected_rev_gained_by_offline_bisection}, 
there exists a bisection $X^*$ satisfying
\[
R(X^*) \geq E_{P_X}[R(X)] \geq \sum_{e \in E} \frac{1}{2} w_e (n - 2y_e^{\pi^*}).
\]

\noindent Thus, by Lemma \ref{lemma.small_weighted_ordering}, 

\begin{align*}
R(X^*) &\geq \sum_{e \in E} \frac{1}{2} w_e (n - 2y_e^{\pi^*}) \\ &=
\sum_{e \in E} \frac{1}{2} w_e (n) - Y_{\pi^*} \\ &\geq
\sum_{e \in E} \frac{1}{2} w_e (n) - \sum_{e \in E} w_e \cdot \frac{|T_e|}{2} \\ &=
\frac{1}{2} \sum_{e \in E} w_e (n - |T_e|) \\ &=
\frac{1}{2}OPT.
\end{align*}

To show that the result is asymptotically tight, 
consider the instance $G = (V,E)$ on $|V| = n$ vertices which is a 
matching of $n/2$ edges, each having weight $1$. In this case, 
one solution is a binary tree in which if $\{i,j\} \in E$ 
then the LCA of $i$ and $j$ (as defined by the solution) 
is $i$ and $j$'s immediate parent. This solution generates 
$n-2$ revenue for every edge of the matching
guaranteeing that $OPT \geq \frac{n}{2}(n-2)$.

On the hand, consider an arbitrary tree of depth 2 such that its first cut is a bisection and denote it by $T$. Let $w_1$ and $w_2$ denote the total weight of edges in each set generated by the cut. Therefore, $w_1 + w_2 \leq \frac{n}{2}$. The revenue generated from $T$ is exactly $\frac{n}{2}w_1 + \frac{n}{2}w_2 \leq \frac{n^2}{4}$. Overall,
\[
R(T) \leq (\frac{1}{2}+o(1))OPT.
\]
\end{proof}

%We note that Theorem \ref{theorem.0.5_bisection} may be %derandomized in the sense that given an optimal (binary) HC %tree we may produce a bisection with half its revenue, %deterministically and in polynomial time.
\begin{remark}
Note that Theorem 2 may be derandomized in the sense that given any binary HC tree we may produce a bisection with at least half its revenue deterministically and efficiently. Indeed, this can be done using the method of conditional expectations. We omit the (simple) details, since although this implies that given an optimal HC tree one can obtain a bisection with revenue at least OPT/2 deterministically and efficiently, this does not yield any new algorithmic consequence when an optimal HC tree is not given.
\end{remark}

\section{A 0.585 Approximation for the Revenue Goal Function}

We define the approximation algorithm as a 2 step process: first we cut all data points using some black box algorithm that produces a bisection, denoted henceforth as $\algmub$. Thereafter, we continue splitting each cluster randomly. 
Denote the combined algorithm by $\algrev$.

In order to show an approximation bound for $\algrev$, we first need to consider the $\mub$ problem. In this problem we are given a weighted graph and our goal is to create a bisection maximizing the weights of uncut edges. We note that if we restrict ourselves to revenue trees which are bisections (i.e., the first cut splits the data points into two equal sets, then each set is cut only once using a "star" subtree - see Figure \ref{figure.example_of_bisection_rev_tree}), the $\mub$ optimal solution and the revenue optimal solution are in fact the same.

\begin{figure}[H]
\centering
\begin{tikzpicture}

\draw [ultra thick] (0,0) -- (1, 1);
\node [circle,draw=black,ultra thick,minimum size=0.8cm] (d_1) at (0,0-.4) {};
\draw [ultra thick] (1,0) -- (1, 1);
\node [circle,draw=black,ultra thick,minimum size=0.8cm] (d_2) at (1,0-.4) {};
\draw [ultra thick] (2,0) -- (1,1);
\node [circle,draw=black,ultra thick,minimum size=0.8cm] (d_3) at (2,0-.4) {};
\draw [ultra thick] (1,1) -- (2.5,2);

\draw [ultra thick] (3,0) -- (4,1);
\node [circle,draw=black,ultra thick,minimum size=0.8cm] (d_4) at (3,0-.4) {};
\draw [ultra thick] (4,0) -- (4,1);
\node [circle,draw=black,ultra thick,minimum size=0.8cm] (d_5) at (4,0-.4) {};
\draw [ultra thick] (5,0) -- (4,1);
\node [circle,draw=black,ultra thick,minimum size=0.8cm] (d_6) at (5,0-.4) {};
\draw [ultra thick] (4,1) -- (2.5,2);

\end{tikzpicture}
\caption{Example of a bisection revenue tree - revenue is only gained through edges uncut by the first cut.}
\label{figure.example_of_bisection_rev_tree}
\end{figure}
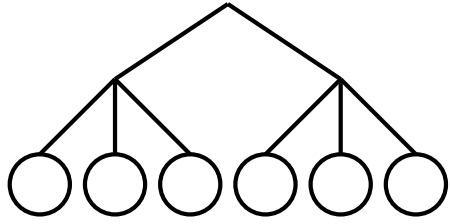

Let $p$ denote the approximation ratio algorithm $\algmub$ guarantees for the $\mub$ problem. Note that $p$ is at least $0.8776$ (see \cite{Better_Balance_by_Being_Biased_A_08776_Approximation_for_Max_Bisection}, \cite{An_improved_semidefinite_programming_hierarchies_rounding_approximation_algorithm_for_maximum_graph_bisection_problems}). Thus, by defining $\algmub$ to be such an algorithm, the following theorem shows that algorithm $\algrev$ is a 0.585 approximation for the revenue problem.

\begin{theorem}
\label{theorem.0.58_approximation}
Algorithm $\algrev$ guarantees an approximation ratio of
$\frac{2}{3}p = 0.585$ for the revenue problem, 
where $p$ is defined as the approximation ratio for the $\mub$ problem.
\end{theorem}

\begin{proof}
Let $T_{\algrev}$ denote the tree generated by algorithm $\algrev$. It is a 
known simple fact that the random algorithm generates a revenue of at least $\frac{1}{3}nw_e $ for any edge $e$. Therefore, if we denote by $W_L$ and $W_R$ the weight of the uncut edges generated by the first cut of our algorithm, then,
\begin{equation*}
R(T_{\algrev}) \geq W_L(\frac{n}{2} + \frac{1}{3}\cdot \frac{n}{2}) + W_R(\frac{n}{2} + \frac{1}{3}\cdot \frac{n}{2}) = (W_L + W_R)\frac{2n}{3}.
\end{equation*}

Denote by $W_{L^*}$ and $W_{R^*}$ the weights of the uncut edges generated by the optimal $\mub$ solution and let $p$ denote our first cut's approximation with respect to the $\mub$ problem. Furthermore, let $X^*$ denote the optimal solution to the revenue problem restricted to bisections. As noted earlier, $X^*$ corresponds to $W_{L^*}$ and $W_{R^*}$. Therefore,
\[
W_L + W_R \geq p(W_{L^*} + W_{R^*}) = \frac{2p}{n}R(X^*).
\]

\noindent Let $OPT$ denote the revenue gained by the optimal solution. Thus, leveraging Theorem \ref{theorem.0.5_bisection} yields,
\begin{align*}
R(T_{\algrev}) \geq& \frac{2n}{3}(W_L + W_R) \geq \\&
\frac{4p}{3} (R(X^*)) \geq \\&
\frac{2p}{3} OPT.
\end{align*}

\noindent Since $p$ is known to be at least $0.8776$, we get that $\algrev$ is a $\frac{2}{3}p = 0.585$ approximation algorithm.
\end{proof}

% \nocite{*}
\bibliographystyle{plain}
\bibliography{bib.bib}
%\pagebreak
%\appendix

%\include{appendix_proofs}
\end{document}